\documentclass[11pt]{article}%
\usepackage{amsmath}
\usepackage{amsfonts}
\usepackage{amssymb}
\usepackage{graphicx}%
\providecommand{\U}[1]{\protect\rule{.1in}{.1in}}
\newtheorem{theorem}{Theorem}

\newtheorem{definition}[theorem]{Definition}

\newtheorem{proposition}[theorem]{Proposition}
\newtheorem{remark}[theorem]{Remark}

\newenvironment{proof}[1][Proof]{\noindent\textbf{#1.} }{\ \rule{0.5em}{0.5em}}
\begin{document}

\title{A Neumann series of Bessel functions representation for solutions of the
radial Dirac system}
\author{Vladislav V. Kravchenko$^{1}$, Elina L. Shishkina$^{2}$ and Sergii M.
Torba$^{1}$\\{\small $^{1}$ Departamento de Matem\'{a}ticas, CINVESTAV del IPN, Unidad
Quer\'{e}taro, }\\{\small Libramiento Norponiente \#2000, Fracc. Real de Juriquilla,
Quer\'{e}taro, Qro., 76230 MEXICO.}\\{\small $^{2}$ Voronezh State University.}\\{\small e-mail: vkravchenko@math.cinvestav.edu.mx,
storba@math.cinvestav.edu.mx, \thanks{Research was supported by CONACYT,
Mexico via the projects 222478 and 284470. Research of Vladislav Kravchenko
was supported by the Regional mathematical center of the Southern Federal
University, Russia.}}}
\maketitle

\begin{abstract}
A new representation for a regular solution of the radial Dirac system of a
special form is obtained. The solution is represented as a Neumann series of
Bessel functions uniformly convergent with respect to the spectral parameter.
For the coefficients of the series convenient for numerical computation
recurrent integration formulas are given. Numerical examples are presented.

\end{abstract}

\section{Introduction}

We consider the one-dimensional radial Dirac system of the form
\begin{equation}%
\begin{pmatrix}
\omega_{1} & -\dfrac{d}{dr}+\dfrac{\kappa}{r}-p(r)\\
\dfrac{d}{dr}+\dfrac{\kappa}{r}-p(r) & \omega_{2}%
\end{pmatrix}%
\begin{pmatrix}
g(r)\\
\\
f(r)
\end{pmatrix}
=0,\label{DiracEq}%
\end{equation}
where $p$ is absolutely continuous complex-valued function on some interval
$[0,b]$, $\omega_{1},\omega_{2}\in\mathbb{C}$, $\kappa$ is the spin-orbit
quantum number, $g$ and $f$ are lower and upper radial wave functions, respectively.

The system \eqref{DiracEq} with $\omega_{1}=\frac{mc^{2}+V_{s}(r)+E-V_{v}%
(r)}{\hbar c}$, $\omega_{2}=\frac{mc^{2}+V_{s}(r)-E+V_{v}(r)}{\hbar c}$ and
$p(r)=\frac{V_{ps}(r)}{\hbar c}$ arises in quantum mechanics when studying
the radial Dirac equation. Here $V_{s}$ is a scalar potential, $V_{v}$ is the
time component of a vector potential and $V_{ps}$ is a pseudoscalar or tensor
potential, see, for example, formula 2.1 in \cite{Alhaidari}, formulas
(21)--(22) in \cite{Eshghi}, system (13)--(14) in \cite{Lisboa} and (1) in
\cite{Linneaus}. The system \eqref{DiracEq} is a special case of the radial
Dirac equation in the presence of a tensor or a pseudoscalar potential, and
when both scalar and vector potentials are constant. System \eqref{DiracEq}
appears in the recent Jackiw-Pi model of the bilayer graphene \cite{Jackiw},
\cite{KhR}. There is a considerable number of publications in which Dirac-type
equations \eqref{DiracEq} are examined, but mostly either exactly solvable
potentials are sought (see, e.g., \cite{Alhaidari}, \cite{Eshghi},
\cite{Lisboa}), or an approximate solution is constructed for a concrete
potential (see, for example, \cite{Ikot}).

In the present work for an arbitrary potential $p(r)$ we obtain an analytical
representation for a regular solution of \eqref{DiracEq} in the form of a
functional series with a simple recurrent integration procedure for
calculating its coefficients. The series has the form of a Neumann series of
Bessel functions (NSBF) (see, e.g., \cite{Watson}, \cite{Wilkins} and
\cite{Baricz et al Book} for more information on NSBF). The following feature
of the obtained NSBF representation makes it especially interesting. Its
partial sums admit spectral parameter independent error estimates, which
guarantee equally accurate approximations of  exact solutions both for small
and for large values of the spectral parameter. More precisely, when the
coupling constants coincide, $\omega_{1}=\omega_{2}$, the estimates are
independent of their values, while in the case $\omega_{1}\neq\omega_{2}$ the
estimates involve the factor $\left\vert \sqrt{\omega_{2}/\omega_{1}%
}\right\vert $, and thus depends on how much the coupling constants differ
from each other.

The NSBF representations for solutions of Sturm-Liouville type equations
proved to be useful for solving both direct and inverse spectral problems
\cite{DKK2019MMAS}, \cite{KKK QuantumFest}, \cite{Kr2019JIIP},
\cite{KrBook2020}, \cite{KNT 2015}, \cite{KrShishkinaTorba2018},
\cite{KrShishkinaTorba2020}, \cite{KT2018Calcolo},
\cite{KT2020improvedNeumann}, \cite{KTC}. In \cite{KNT 2015} an NSBF
representation was obtained for solutions of the one-dimensional stationary
Schr\"{o}dinger equation. In \cite{KT2018Calcolo} that result was generalized
onto the case of an arbitrary regular Sturm-Liouville equation. Recently in
\cite{KT2020improvedNeumann} an NSBF representation was obtained for regular
solutions of perturbed Bessel equations. In \cite{DKK2019MMAS}, \cite{KKK
QuantumFest}, \cite{Kr2019JIIP}, \cite{KrBook2020},
\cite{KrShishkinaTorba2020} NSBF representations for solutions were used for
solving inverse spectral problems.

In the present paper an NSBF representation for regular solutions of
\eqref{DiracEq} is obtained by transforming the system into a couple of
perturbed Bessel equations and using results from \cite{KT2020improvedNeumann}%
. We prove the above mentioned error estimates for partial sums of the series
representations and discuss the numerical implementation of the NSBF
representation. We show that the spectral parameter independent error
estimates are evident, indeed, in numerical experiments and show the
applicability of the obtained NSBF representation for solving spectral
problems for \eqref{DiracEq}.

The paper is organized as follows. In Section \ref{Sect2} we obtain the NSBF
representation for the regular solution of \eqref{DiracEq} and prove a
convergence result for the approximate solution. In Section \ref{Sect3} we
summarize the steps required for numerical solution of equation
\eqref{DiracEq} and related spectral problems using the proposed
representation and show numerical results for the Dirac oscillator.

\section{A representation of the solution}

\label{Sect2}

Consider the following two component radial Dirac system
\begin{align}
\left(  \frac{d}{dr}-\frac{\kappa}{r}+p(r)\right)  f & =\omega_{1}%
g,\label{EQ01}\\
\left(  \frac{d}{dr}+\frac{\kappa}{r}-p(r)\right)  g & =-\omega_{2}f,
\label{EQ02}%
\end{align}
where $\omega_{1},\omega_{2}\in\mathbb{C}$, $\kappa\geq\frac{1}{2}$, and
$p\in\operatorname{AC}[0,b]$ is in general a complex valued function.

\begin{definition}
A pair of functions $(f_{\kappa},g_{\kappa})$ is called a \textbf{regular
solution} of the system \eqref{EQ01}--\eqref{EQ02} if it satisfies the system
as well as the following asymptotic conditions
\[
f_{\kappa}(r)\sim C_{f}r^{\kappa},\qquad g_{\kappa}(r)\sim C_{g}r^{\kappa
+1},\quad\text{when }r\rightarrow0
\]
where $C_{f}$ and $C_{g}$ are some constants.
\end{definition}

Together with the potential $p$ the following functions will be considered
\begin{equation}
q_{1}(r)=p\,^{\prime}(r)-\frac{2\kappa}{r}p(r)+p^{2}(r)\qquad\text{and}\qquad
q_{2}(r)=-p^{\prime}(r)-\frac{2\kappa}{r}\,p(r)+p^{2}(r). \label{Q}%
\end{equation}
Note that if $(f_{\kappa},g_{\kappa})$ is a regular solution of
\eqref{EQ01}--\eqref{EQ02}, the functions $f_{\kappa}$ and $g_{\kappa}$ are
necessarily regular solutions of the equations
\begin{equation}
-f^{\prime\prime}+\left(  \frac{\kappa(\kappa-1)}{r^{2}}+q_{2}(r)\right)
f=\omega^{2}f,\qquad r\in(0,b] \label{Eqq2}%
\end{equation}
and
\begin{equation}
-g^{\prime\prime}+\left(  \frac{\kappa\left(  \kappa+1\right)  }{r^{2}}%
+q_{1}(r)\right)  g=\omega^{2}g,\qquad r\in(0,b], \label{Eqq1}%
\end{equation}
respectively with $\omega^{2}=\omega_{1}\omega_{2}$.

Note that for $p\in\operatorname{AC}[0,b]$ both potentials $q_{1}$ and $q_{2}$
are such that $r^{\varepsilon}q_{1,2}(r)\in L_{1}(0,b)$ for any small
$\varepsilon>0$, hence the conditions on the potential from
\cite{KT2020improvedNeumann} are satisfied. In order to apply the results of
\cite{KT2020improvedNeumann} to equations \eqref{Eqq2} and \eqref{Eqq1} we
need two solutions of the equations
\begin{equation}
-f_{0}^{\prime\prime}+\left(  \frac{\kappa(\kappa-1)}{r^{2}}+q_{2}(r)\right)
f_{0}=0\qquad\text{and}\qquad-g_{0}^{\prime\prime}+\left(  \frac{\kappa
(\kappa+1)}{r^{2}}+q_{1}(r)\right)  g_{0}=0,\label{Bessel hom}%
\end{equation}
non-vanishing on $(0,b]$ and satisfying the following asymptotics at zero
\begin{equation}
f_{0}(r)\sim r^{\kappa}\qquad\text{and}\qquad g_{0}(r)\sim r^{\kappa+1}%
,\quad\text{when }r\rightarrow0.\label{asympt f0 g0}%
\end{equation}
The solution $f_{0}$ can be directly obtained by taking $\omega_{1}=0$ in
\eqref{EQ01} and is given by
\begin{equation}
f_{0}(r)=r^{\kappa}\exp\left(  -\int_{0}^{r}p(s)\,ds\right)  .\label{Sol f0}%
\end{equation}
To obtain the solution $g_{0}$ note that the function $1/f_{0}$ is a solution
of \eqref{EQ02} with $\omega_{2}=0$,\ and hence it is the solution of the
second equation in \eqref{Bessel hom} satisfying the asymptotic relation
$1/f_{0}(r)\sim r^{-\kappa}$, $r\rightarrow0$. A second linearly independent
solution of \eqref{EQ02} with $\omega_{2}=0$ can be chosen in the form
$\frac{C}{f_{0}(r)}\int_{0}^{r}f_{0}^{2}(s)\,ds$. Chosing $C=2\kappa+1$, i.e.,
taking
\begin{equation}
g_{0}(r)=(2\kappa+1)r^{-\kappa}\exp\left(  \int_{0}^{r}p(s)\,ds\right)
\int_{0}^{r}t^{2\kappa}\exp\left(  -2\int_{0}^{t}p(s)\,ds\right)
\,dt\label{Sol g0}%
\end{equation}
we obtain the solution of the second equation in \eqref{Bessel hom} satisfying \eqref{asympt f0 g0}.

It can be seen from \eqref{Sol f0} and \eqref{Sol g0} that the derivatives of
the solutions $f_{0}$ and $g_{0}$ are given by
\begin{equation}
\label{Sols derivs}f_{0}^{\prime}(r)=\left( \frac{\kappa}r-p(r)\right)
f_{0}(r)\qquad\text{and}\qquad g_{0}^{\prime}(r) = \left( p-\frac{\kappa
}r\right) g_{0}(r) + (2\kappa+1)f_{0}(r).
\end{equation}

The solution $f_{0}$ given by \eqref{Sol f0} is always non-vanishing on
$(0,b]$. The solution $g_{0}$ given by \eqref{Sol g0} is definitely
non-vanishing for real valued potentials $p$ but may possess zeros for complex
valued functions $p$. For this reason the following \textbf{assumption (A)}
concerning the potential $p$ will be made throughout this paper. We assume
that the second equation in \eqref{Bessel hom} admits a regular solution
$g_{0}$ which does not vanish on $(0,b]$. This assumption does not imply any
additional restriction on $p$ for the following reason. In \cite[Proposition
B.1]{KT2020improvedNeumann} we show that one can always choose such a constant
$c$ that the second equation in \eqref{Bessel hom} with the potential
$\widetilde{q}_{1}(x):=q_{1}(x)+c$ possesses a non-vanishing solution.
Equation \eqref{Eqq1} can then be written as $-g^{\prime\prime}+\left(
\frac{\kappa\left(  \kappa+1\right)  }{r^{2}}+\widetilde{q}_{1}(r)\right)
g=\widetilde{\omega}^{2}g$ with $\widetilde{\omega}^{2}=\omega^{2}+c$ which
leads to the same results and conclusions as below. Also, one may construct
the regular solution of the system \eqref{EQ01}--\eqref{EQ02} using only the
solution $f$ and its derivative (see Remark \ref{Remark OtherWaySol}), however
loosing an attractive possibility to verify the accuracy of approximate
solutions (see Remark \ref{Remark Accuracy} and Subsection
\ref{SubSectAlgorith}).

Thus, without loss of generality we assume that the regular solution $g_{0}$
of the second equation in \eqref{Bessel hom} satisfying \eqref{asympt f0 g0}
does not have zeros in $(0,b]$.

\begin{theorem}
\label{ThMain}Let $p\in\operatorname{AC}[0,b]$, and the assumption (A) be
fulfilled. Then a regular solution of the system \eqref{EQ01}--\eqref{EQ02}
satisfying the asymptotic relations (here $\omega^{2}=\omega_{1}\omega_{2}$)
\[
f_{\kappa}(r)\sim-\frac{\omega^{\kappa+1}}{\omega_{2}} d(\kappa-1)r^{\kappa
}\qquad\text{and}\qquad g_{\kappa}(r)\sim\omega^{\kappa+1}d(\kappa
)r^{\kappa+1},\quad r\rightarrow0,
\]
has the form
\begin{align}
\label{Sol1}f_{\kappa}(r) & =-\frac{\omega^{2}}{\omega_{2}} r j_{\kappa
-1}(\omega r)-\frac{\omega}{\omega_{2}}\sum_{n=0}^{\infty}\beta_{2,n}%
(r)j_{\kappa+2n}(\omega r),\\
g_{\kappa}(r) & =\omega r j_{\kappa}(\omega r)+\sum_{n=0}^{\infty}\beta
_{1,n}(r)j_{\kappa+2n+1}(\omega r),\label{Sol2}%
\end{align}
where $j_{\nu}(r)=\sqrt{\frac{\pi}{2r}}J_{\nu+{\frac{1}{2}}}(r)$ is the
spherical Bessel function of the first kind,
\[
d(\kappa):=\frac{\sqrt\pi}{2^{\kappa+1}\Gamma(\kappa+3/2)}.
\]
Denote $u_{1}:=g_{0}$ and $u_{2}:=f_{0}$, where $f_{0}$ and $g_{0}$ are
solutions of \eqref{Bessel hom} satisfying \eqref{asympt f0 g0}. Then the
functions $\beta_{j,n}$, $j\in\{1,2\}$, $n\ge0$, can be found from the
recurrent formulas
\begin{align}
\beta_{j,0}(r)  & = (2\kappa-2j + 5)\left( \frac{u_{j}(r)}{r^{\kappa+2-j}%
}-1\right) ,\qquad j\in\{1,2\},\label{Beta1}\\
\beta_{j,n}(r) & =-\frac{4n+2\kappa-2j+5}{4n+2\kappa-2j+1}\left[
\beta_{j,n-1}(r)+ \frac{2(4n+2\kappa-2j+3)u_{j}(r)\theta_{j,n}(r)}%
{r^{2n+\kappa-j+2}}\right]  ,\label{RecBeta1}\\
\displaybreak[2] \theta_{j,n}(r) & =\int_{0}^{r}\frac{\eta_{j,n}%
(t)-t^{2n+\kappa-j+1}\beta_{j,n-1}(t)u_{j}(t)}{u_{j}^{2}(t)}dt,\\
\displaybreak[2] \eta_{j,n}(r) & =\int_{0}^{r}\left[  tu_{j}^{\prime
}(t)+(2n+\kappa-j+1)u_{j}(t)\right]  t^{2n+\kappa-j}\beta_{j,n-1}%
(t)dt.\label{etan}%
\end{align}

\end{theorem}

\begin{proof}
From (\ref{EQ02}) we have that
\[
f_{\kappa}=-\frac{1}{\omega_{2}}\left(  g_{\kappa}^{\prime}+\kappa g_{\kappa
}/r-pg_{\kappa}\right)  .
\]
Hence if $g_{\kappa}(r)\sim d(\kappa)(\omega r)^{\kappa+1}$, $g^{\prime
}_{\kappa}(r)\sim(\kappa+1)\omega d(\kappa)(\omega r)^{\kappa}$ when
$r\rightarrow0$, then $f_{\kappa}(r)\sim-\frac{\omega}{\omega_{2}}\left(
2\kappa+1\right) d(\kappa)(\omega r)^{\kappa}$. Now we apply Theorem 5.2 from
\cite{KT2020improvedNeumann} in order to find out that a solution $g_{\kappa}$
of (\ref{Eqq1}) satisfying the relation $g_{\kappa}(r)\sim d(\kappa)(\omega
r)^{\kappa+1}$, has the form (\ref{Sol2}), meanwhile a solution $\widetilde
{f}_{\kappa}$ of (\ref{Eqq2}) satisfying the relation $\widetilde{f}_{\kappa
}(r)\sim d(\kappa-1)(\omega r)^{\kappa}$, when $r\rightarrow0$, can be written
as
\[
\widetilde{f}_{\kappa}(r)=\omega r j_{\kappa-1}(\omega r)+\sum_{n=0}^{\infty
}\beta_{2,n}(r)j_{\kappa+2n}(\omega r).
\]
And $f_{\kappa}=-\frac{\omega}{\omega_{2}}\left(  2\kappa+1\right)
\frac{d(\kappa)}{d(\kappa-1)} \widetilde{f}_{\kappa} = -\frac{\omega}%
{\omega_{2}} \widetilde{f}_{\kappa}$, which leads to (\ref{Sol1}).
\end{proof}

For practical use of the representation (\ref{Sol1}), (\ref{Sol2}) the
estimates of the difference between the exact solution and its approximation
defined as
\begin{align}
f_{\kappa,N}(r) &  =--\frac{\omega^{2}}{\omega_{2}}rj_{\kappa-1}(\omega
r)-\frac{\omega}{\omega_{2}}\sum_{n=0}^{N}\beta_{2,n}(r)j_{\kappa+2n}(\omega
r),\label{fkN}\\
g_{\kappa,N}(r) &  =\omega rj_{\kappa}(\omega r)+\sum_{n=0}^{N}\beta
_{1,n}(r)j_{\kappa+2n+1}(\omega r)\label{gkN}%
\end{align}
are needed. From Theorem \ref{ThMain} using \cite[Theorem 5.2]%
{KT2020improvedNeumann} the following result follows immediately.

\begin{proposition}
Under the conditions of Theorem \ref{ThMain} the following inequalities are
valid
\[
|g_{\kappa}(r)-g_{\kappa,N}(r)|\leq\sqrt{r}\varepsilon_{N}(r)\quad
\text{and\quad}|f_{\kappa}(r)-f_{\kappa,N}(r)|\leq\left| \frac{\omega}%
{\omega_{2}}\right| \sqrt{r} \varepsilon_{N}(r)
\]
for all $\omega\in\mathbb{R}$, $\omega_{2}\in\mathbb{R}\setminus\{0\}$, where
$\varepsilon_{N}$ is a nonnegative function independent on $\omega_{1}$ and
$\omega_{2}$, such that $\max_{r\in[0,b]}\varepsilon_{N}(r)\to0$ as
$N\to\infty$. Similar result holds for $\omega$ belonging to a strip
$|\operatorname{Im}\omega|\le C$, with addition of a multiplicative constant
dependent only on the value of $C$.

Suppose additionally that $p\in W_{1}^{2k}[0,b]$, $p(0)=0$ and $\frac{p(r)}%
{r}\in W_{1}^{2k-1}[0,b]$ for some $k\in\mathbb{N}$. Here $W_{1}^{k}[0,b]$
denotes the class of functions having $k$ derivatives, the last one belonging
to $L_{1}[0,b]$ space, and $p(r)/r$ is assumed to have a finite limit as
$r\rightarrow0$. Then there exists a constant $c$, such that
\[
\varepsilon_{N}(r)\leq\frac{c}{N^{k}},\qquad2N>\kappa+k+1.
\]

\end{proposition}

The independence of $\varepsilon_{N}$ of $\omega$ implies that the approximate
solution $(f_{\kappa,N},g_{\kappa,N})$ remains good even for very large values
of $\left\vert \operatorname{Re}\omega\right\vert $.

\begin{remark}
\label{Remark Accuracy} Even though the derivatives of the regular solutions
$(f_{\kappa}, g_{\kappa})$ are readily available from \eqref{EQ01} and
\eqref{EQ02}, an independent representation (useful, e.g., for verification of
accuracy of approximate solutions) for them can be obtained from \cite[Theorem
6.3]{KT2020improvedNeumann}. Under the conditions and notations of Theorem
\ref{ThMain} let $Q_{j}(r):=\int_{0}^{r} q_{j}(t)\,dt$, $j\in\{1,2\}$ and let
the functions $\gamma_{j,n}$, $j\in\{1,2\}$, $n\ge0$ be defined as
\begin{align*}
\gamma_{j,0}(r)  &  = (2\kappa-2j + 5) \left( \frac{u_{j}^{\prime}%
(r)}{r^{\kappa-j+2}} - \frac{\kappa-j +2}{r} - \frac{Q_{j}(r)}2\right) ,\\
\gamma_{j,n}(r)  & = -\frac{4n+2\kappa-2j+5}{4n+2\kappa-2j+1}\biggl[\gamma
_{j,n-1}(r)\\
& \quad+ (4n+2\kappa-2j+3)\left( \frac{2u_{j}^{\prime}(r) \theta_{j,n}%
(r)}{r^{2n+\kappa-j+2}} + \frac{2\eta_{j,n}(r)}{u_{j}(r)r^{2n+\kappa-j+2}} -
\frac{\beta_{j,n-1}(r)}r\right) \biggr].
\end{align*}
Then
\begin{align}
f_{\kappa}^{\prime}(r)  & = -\frac{\omega^{3}}{\omega_{2}} r j_{\kappa
-2}(\omega r) -\left( \frac{rQ_{2}(r)}2 - \kappa+1\right) \frac{\omega^{2}%
}{\omega_{2}} j_{\kappa-1}(\omega r) - \frac{\omega}{\omega_{2}} \sum
_{n=0}^{\infty}\gamma_{2,n}(r)j_{2n+\kappa}(\omega r),\label{fprime}\\
g_{\kappa}^{\prime}(r) &= \omega^2 r j_{\kappa-1}(\omega r) + \left( \frac{rQ_{1}(r)}2 -
\kappa\right) \omega j_{\kappa}(\omega r) +\sum_{n=0}^{\infty}\gamma
_{1,n}(r)j_{2n+\kappa+1}(\omega r).\label{gprime}%
\end{align}

\end{remark}

\begin{remark}
\label{Remark OtherWaySol} The regular solution of the system
\eqref{EQ01}--\eqref{EQ02} can be obtained using only the particular solution
$f_{0}$ and related functions $\beta_{2,n}$ and $\gamma_{2,n}$, $n\ge0$,
without the need of the functions $g_{0}$, $\beta_{1,n}$ and $\gamma_{1,n}$ at
all. Indeed,
\[
g_{\kappa}= \frac{1}{\omega_{1}}\left( f_{\kappa}^{\prime}-\frac\kappa r
f_{\kappa}+p(r)f_{\kappa}\right) \qquad\text{and}\qquad g_{\kappa}^{\prime}=
-\omega_{2} f_{\kappa}- \frac\kappa r g_{\kappa}+ p(r) g_{\kappa},
\]
and the representations for $f_{\kappa}$ and $f_{\kappa}^{\prime}$ are given
by \eqref{Sol1} and \eqref{fprime}.
\end{remark}

\section{Numerical results}

\label{Sect3}

\subsection{Description of the algorithm}

\label{SubSectAlgorith} A numerical method based on the representation
\eqref{Sol1}--\eqref{Sol2} of the regular solution of the system
\eqref{EQ01}--\eqref{EQ02} consists in the following steps.

\begin{enumerate}
\item Compute a pair $(f_{0}, g_{0})$ of regular solutions of
\eqref{Bessel hom} satisfying \eqref{asympt f0 g0} using \eqref{Sol f0} and
\eqref{Sol g0}. Compute also their derivatives $(f_{0}^{\prime}, g_{0}%
^{\prime})$. In the case that the coefficient $p$ is complex valued, check if
the assumption (A) holds, and if not, proceed as described in Remark
\ref{Remark OtherWaySol} or look for a spectral shift (see Appendix B in
\cite[(8.1)]{KT2020improvedNeumann}) such that a pair of solutions $(f_{0},
g_{0})$ becomes non-vanishing.

\item Compute the coefficients $\beta_{j,n}$, $j\in\{1,2\}$, $n\in
\{0,1,\ldots,N\}$ using the formulas \eqref{Beta1}--\eqref{etan}. Note that
the coefficients $\beta_{j,n}$ satisfy \cite{KT2020improvedNeumann}
\begin{equation}
\label{betak verif}\sum_{n=0}^{\infty}(-1)^{n}\beta_{j,n}(r)=\frac{rQ_{j}%
(r)}2,\qquad r\in[0,b],\ j\in\{1,2\}
\end{equation}
and decay to zero (however, not necessary monotonously) as $n\to\infty$. The
equality \eqref{betak verif} can be used to estimate an optimal number of the
coefficients $N$, as a value where the truncated sums cease to decrease when
$N$ increases.

\item Compute approximate solutions $f_{\kappa, N}$ and $g_{\kappa,N}$ using
\eqref{Sol1} and \eqref{Sol2}.

\item The accuracy of the obtained approximations can be estimated by
calculating the discrepancies
\begin{equation}
f_{\kappa,N}^{\prime}-\frac{\kappa}{r}f_{\kappa,N}+p(r)f_{\kappa,N}-\omega
_{1}g_{\kappa,N}\qquad\text{and}\qquad g_{\kappa,N}^{\prime}+\frac{\kappa}%
{r}f_{\kappa,N}-p(r)f_{\kappa,N}+\omega_{2}g_{\kappa,N},\label{Eq discr}%
\end{equation}
where $f_{\kappa,N}^{\prime}$ and $g_{\kappa,N}^{\prime}$ are computed from
the truncated series \eqref{fprime} and \eqref{gprime}.
\end{enumerate}

We refer the reader to \cite{KNT 2015} and \cite{KTC} for implementation
details of the proposed algorithm.

\subsection{The Dirac oscillator}

As a test example for the proposed algorithm we consider the Dirac oscillator
\cite{MoSzcz1989, BMNS1990, DoGo1990}.

The large radial component $F(r)$ and the small radial component $G(r)$ of the
Dirac wave function are solutions of the following system
\begin{align}
\left( -\frac{d}{dr}+\left( \frac{\varepsilon(j+1/2)}{r}+m\omega r\right)
\right) G(r)  &  = (E-m)F(r),\label{ExEqDirac1}\\
\left( \frac{d}{dr}+\left( \frac{\varepsilon(j+1/2)}{r}+m\omega r\right)
\right) F(r)  &  = (E+m)G(r),\label{ExEqDirac2}%
\end{align}
where $j$ is the total angular momentum quantum number, $\varepsilon=\pm1$,
$m$ is the mass of the particle and $\omega$ is the frequency. Note that the
number
\[
l:=j+\frac\varepsilon2
\]
is the orbital momentum quantum number and is an integer number, i.e., the
fractional part of $j$ is always equal to $1/2$.

\begin{figure}[tbh!]
\centering
\begin{tabular}
[c]{cc}%
\includegraphics[bb=0 0 216 144
height=2n,
width=3in
]{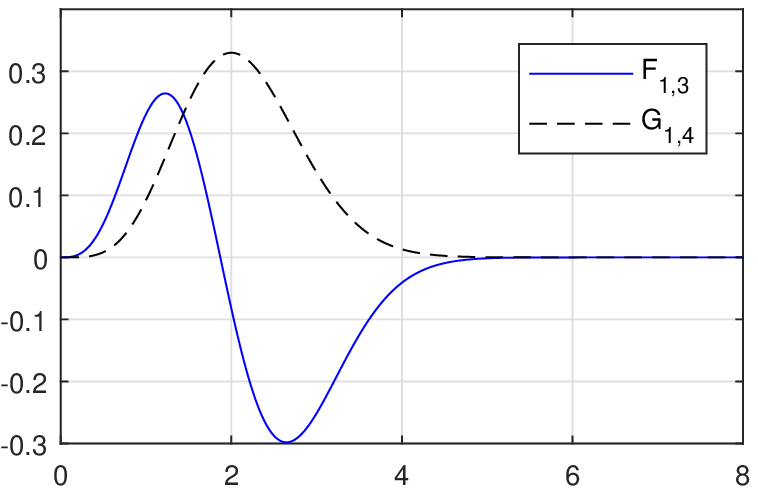} & \includegraphics[bb=0 0 216 144
height=2n,
width=3in
]{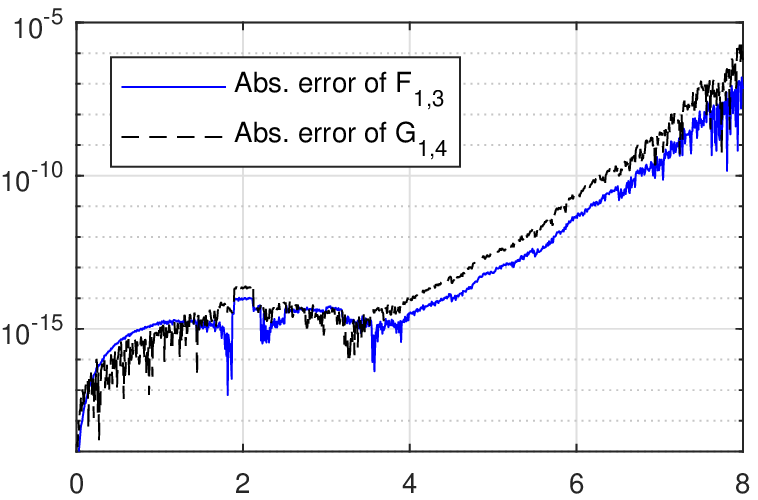}\\
\medskip\includegraphics[bb=0 0 216 144
height=2n,
width=3in
]{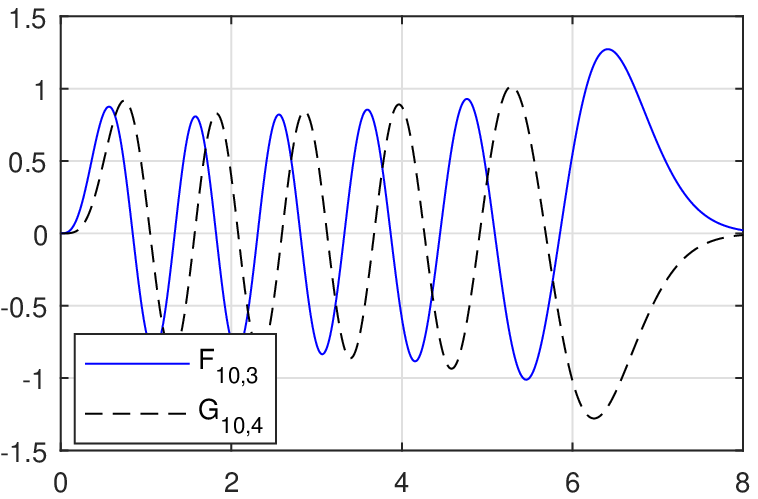} & \includegraphics[bb=0 0 216 144
height=2n,
width=3in
]{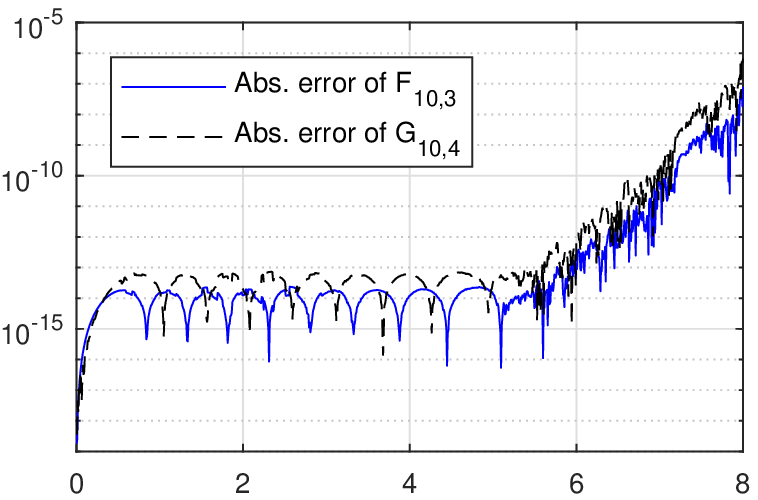}\\
\medskip\includegraphics[bb=0 0 216 144
height=2n,
width=3in
]{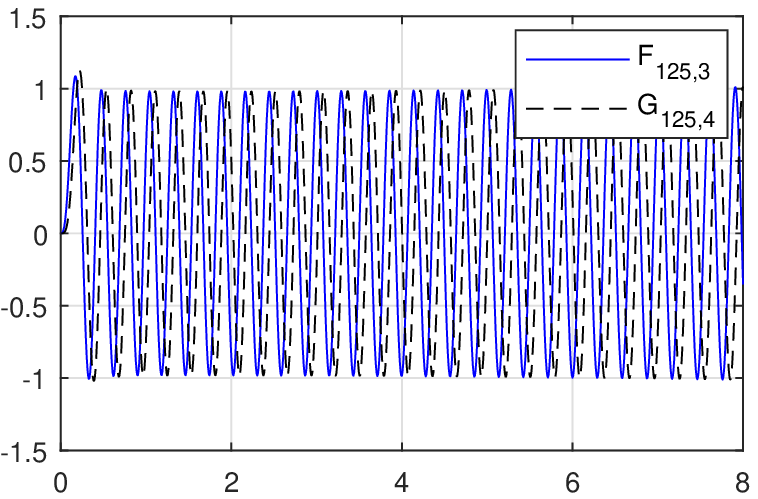} & \includegraphics[bb=0 0 216 144
height=2n,
width=3in
]{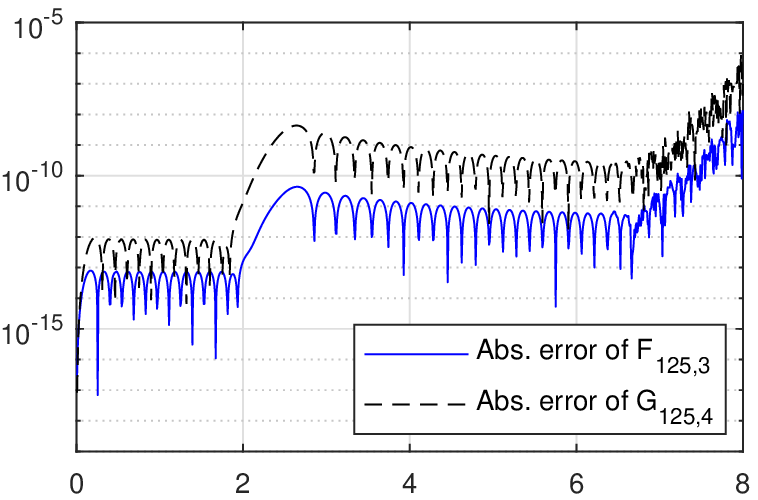}
\end{tabular}
\caption{On the left: components $F_{n,3}$ and $G_{n,4}$ of the eigenfunction
of the Dirac oscillator for $n\in\{1,10,125\}$ with the parameters $j=5/2$,
$\varepsilon=-1$ and $m=\omega=1$. On the right: absolute errors of these
components.}%
\label{Figure Eigs}%
\end{figure}

The energy spectrum can be obtained \cite{BMNS1990} from
\[
E^{2}-m^{2} = m\omega\bigl(2(N+1)+\varepsilon(2j+1)\bigr)
\]
for the positive-energy states, and from
\[
E^{2}-m^{2} = m\omega\bigl(2(N+2)+\varepsilon(2j+1)\bigr)
\]
for the negative-energy states. Here $N=2n+l$, $n=0,1,2,\ldots$, is the
principal quantum number. The corresponding eigenfunctions are given by
\begin{align}
F_{n,l}(r)  & = A\left(  r\sqrt{m\omega}\right) ^{l+1} \exp(-m\omega r^{2}/2)
L_{n}^{l+1/2}(m\omega r^{2}),\label{ExEigF}\\
G_{n,l-\varepsilon}(r)  & = A\left(  r\sqrt{m\omega}\right) ^{l+1-\varepsilon}
\exp(-m\omega r^{2}/2) L_{n+\varepsilon/2-1/2}^{l-\varepsilon+1/2}(m\omega
r^{2}),\label{ExEigG}%
\end{align}
where $L_{k}^{s}(x)$ is an associated Laguerre polynomial.

The system \eqref{ExEqDirac1}--\eqref{ExEqDirac2} is of the type considered in
this paper. Since the potential of the problem is increasing, we approximated
the semiaxis spectral problem (of finding the values of $E$ for which the
regular solution belongs to $L_{2}(0,\infty)$) by truncating the potential and
considering the Dirichlet boundary condition. For any non-trivial solution
both $f_{\kappa}$ and $g_{\kappa}$ can not be equal to zero at one point, so
we choose the function $f_{\kappa}$ and considered
\[
f_{\kappa}(B)=0
\]
as the boundary condition for the problem truncated onto $[0,B]$ segment. We
refer the reader to \cite[Section 7.4]{PryceBook} for additional details on
the convergence of the eigenvalues of truncated problems to the exact ones.

\begin{figure}[tbh]
\centering
\begin{tabular}
[c]{cc}%
$\varepsilon=1$ & $\varepsilon=-1$\\
\includegraphics[bb=0 0 216 173
height=2.4n,
width=3in
]{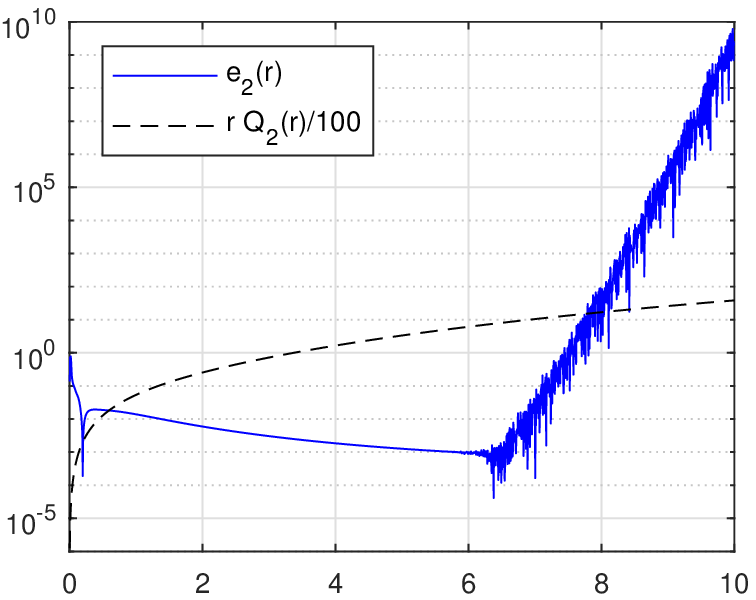} & \includegraphics[bb=0 0 216 173
height=2.4n,
width=3in
]{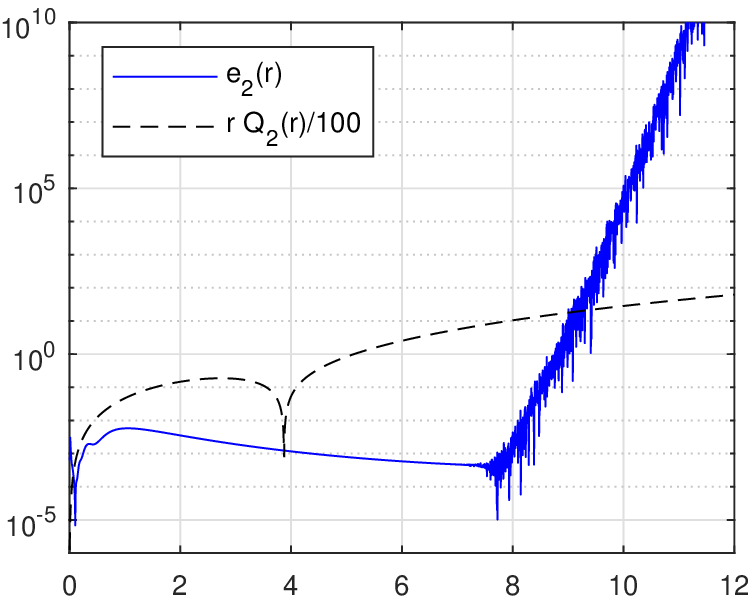}
\end{tabular}
\caption{Application of formula \eqref{betak verif} to determine optimal
truncation interval for the Dirac oscillator with the parameters $j=5/2$ and
$m=\omega=1$. On both plots black dashed line shows the value of
$|rQ_{2}(r)/100|$, and the blue solid line shows $e_{2}(r) := \min
_{N(r)\le100}\bigl|\sum_{n=0}^{N(r)}(-1)^{n} \beta_{2,n}(r) - rQ_{2}%
(r)/2\bigr|$ for computed coefficients $\beta_{2,n}$.}%
\label{Figure Crop}%
\end{figure}

All the computations were performed in machine precision using Matlab 2017. We
refer the reader to \cite{KTC} for the details of the numerical realization.
We considered two sets of parameters, having $\varepsilon=\pm1$ and in both
$j=5/2$ and $m=\omega=1$. For $\varepsilon=1$ the corresponding potential is
$p=-m\omega r$ in the notations of \eqref{EQ01}, \eqref{EQ02}, and for
$\varepsilon=-1$ the corresponding potential is $p=m\omega r$. In the first
case the corresponding particular solution $f_{0}$ given by \eqref{Sol f0} is
rapidly increasing, for the second case $f_{0}$ is rapidly decreasing. We
decided to not implement interval subdivision techniques, and utilize the
proposed representation directly to illustrate that even straightforward
implementation can deliver highly accurate results.

First, we compare approximate solutions with the exact ones for the case
$\varepsilon=-1$ for three eigenvalues $E^{2}-m^{2}\in\{4,40,500\}$,
corresponding to $n\in\{1,10,125\}$. In terms of the system \eqref{EQ01},
\eqref{EQ02} we have taken $\omega_{1}=2$, $\omega_{2}\in\{2,20,250\}$. On
Figure \ref{Figure Eigs} we present the solutions and corresponding absolute
errors. As one can observe, the error does not increase for large values of
$\omega$ (corresponding to higher index eigenfunctions) and only increases for
large values of $r$ due to machine precision limitations. The approach
presented in Remark \ref{Remark OtherWaySol} delivered a more accurate
solution component $G$. This is due to the error near $r=0$ in the particular
solution $g_{0}$ computed by \eqref{Sol g0}. For that reason on the plots we
present the absolute errors obtained with the aid of the formula from Remark
\ref{Remark OtherWaySol}.

Approximate solution of the spectral problem requires truncating the interval.
A larger interval allows one to compute more eigenvalues and more accurately.
However this leads to larger errors in all the coefficients $\beta_{j,n}$
computed, due to machine precision limitations. The equality
\eqref{betak verif} can be utilized to estimate automatically a truncation
parameter $B$. We took the segment $[0,20]$, represented all the functions
involved by 100001 uniformly spaced on $[0,20]$ points and computed 100
coefficients $\beta_{2,n}$. After that we checked for each $r$ the convergence
of partial sums in \eqref{betak verif} to $rQ_{2}(r)/2$. Due to machine
precision limitations, the difference between $\sum_{n=0}^{N}(-1)^{n}%
\beta_{2,n}(r)$ and $rQ_{2}(r)/2$ reaches a plateau at some particular value
of $N(r)$, meaning that the difference essentially does not decrease anymore
when $N$ increases. Let $e_{2}(r):=\bigl|\sum_{n=0}^{N(r)}(-1)^{n}\beta
_{2,n}(r)-rQ_{2}(r)/2\bigr|$. We chose as the truncation parameter $B$ the
value $0.99\cdot r_{0}$, where $r_{0}$ is such that for all $r<r_{0}$ the
value $e_{2}(r)$ is small in comparison with $rQ_{2}(r)$ (to be more precise,
$e_{2}(r)<|rQ_{2}(r)|/100$), but for $r>r_{0}$ the error $e_{2}(r)$ can be
larger than $|rQ_{2}(r)|/100$. As a result, $B=7.4786$ was chosen for
$\varepsilon=1$, and $B=9.0168$ was chosen for $\varepsilon=-1$. See Figure
\ref{Figure Crop} illustrating this procedure.

In Table \ref{Ex1Table1} we present approximate eigenvalues $E^{2}-m^{2}$ for
the parameters $\varepsilon=\pm1$, $j=5/2$, $m=\omega=1$ computed on the
truncated intervals $[0,B]$.

\begin{table}[tbh]
\centering
\begin{tabular}
[c]{cc}%
$\varepsilon= 1$, on $[0, 7.4786]$ & $\varepsilon= -1$, on $[0, 9.0168]$\\%
\begin{tabular}
[c]{|c|c|}\hline
Exact $E^{2}-m^{2}$ & Approximate\\\hline
14 & 13.999999999987\\
18 & 17.9999999998183\\
22 & 21.9999999982828\\
26 & 25.9999999642871\\
30 & 29.9999994276682\\
34 & 33.9999942694057\\
38 & 37.9999616653564\\
42 & 42.0001005681044\\
46 & 46.000048366715\\
50 & 50.0125330275323\\
54 & 54.0378367431326\\
58 & 58.2112436119225\\\hline
Number of $\beta_{2,n}$ used & 24\\\hline
\end{tabular}
&
\begin{tabular}
[c]{|c|c|}\hline
Exact $E^{2}-m^{2}$ & Approximate\\\hline
0 & $7.8\cdot10^{-32}$\\
4 & 3.99999999999999\\
8 & 7.99999999999994\\
12 & 12.0000000000002\\
16 & 16.0000000000035\\
20 & 20.0000000000206\\
24 & 24.0000000000766\\
28 & 27.9999999997443\\
32 & 32.0000000015208\\
36 & 35.9999999944383\\
40 & 39.9999997918537\\
44 & 44.0000015951936\\
48 & 48.0000051599966\\\hline
Number of $\beta_{2,n}$ used & 29\\\hline
\end{tabular}
\end{tabular}
\caption{The eigenvalues for the Dirac oscillator problem \eqref{ExEqDirac1},
\eqref{ExEqDirac2} truncated onto the segment $[0,B]$. Parameters used:
$\varepsilon=\pm1$, $j=5/2$, $m=\omega=1$. The last line shows the number of
terms used in approximate solution \eqref{fkN}.}%
\label{Ex1Table1}%
\end{table}

\end{document}